\documentclass[letterpaper,11pt]{article}
\clearpage{}\usepackage{fullpage}
\usepackage{amsmath,amsthm,amsfonts,dsfont}
\usepackage{amssymb,latexsym,graphicx}
\usepackage{mathtools}
\usepackage{hyperref}
\usepackage{xcolor}
\hypersetup{
	colorlinks,
	linkcolor={red!75!black},
	citecolor={blue!75!black},
	urlcolor={blue!75!black}
}

\newtheorem{theorem}{Theorem}[section]

\newtheorem{lemma}[theorem]{Lemma}

\newtheorem{fact}[theorem]{Fact}

\newtheorem{corollary}[theorem]{Corollary}
\newtheorem{definition}[theorem]{Definition}

\newcommand{\R}{\ensuremath{\mathbb{R}}}
\newcommand{\Z}{\ensuremath{\mathbb{Z}}}

\newcommand{\lat}{\mathcal{L}}

\newcommand{\eps}{\varepsilon} 
\renewcommand{\epsilon}{\varepsilon}
\newcommand{\poly}{\mathrm{poly}}

\newcommand{\vol}{\mathrm{vol}}

\DeclareMathOperator{\spn}{span}

\renewcommand{\vec}[1]{\ensuremath{\boldsymbol{#1}}}
\newcommand{\basis}{\ensuremath{\mathbf{B}}}

\DeclarePairedDelimiter\ceil{\lceil}{\rceil}

\mathtoolsset{centercolon}\clearpage{}

\usepackage[utf8]{inputenc}
\usepackage{mathbbold}
\usepackage[all]{xy}
\usepackage{amsmath}
\usepackage{amsthm}
\usepackage{amsfonts}
\usepackage{appendix}
\usepackage{pifont} 
\usepackage{amssymb}
\usepackage{mathrsfs}
\usepackage[T1]{fontenc}
\usepackage{ae}
\usepackage{boxedminipage}
\usepackage{bibspacing}
\usepackage{tweaklist}
\renewcommand{\enumhook}{\setlength{\topsep}{0pt}
  \setlength{\itemsep}{0pt}}
  \usepackage{tweaklist}
\renewcommand{\itemhook}{\setlength{\topsep}{0pt}
  \setlength{\itemsep}{0pt}}
  \usepackage{tweaklist}
\renewcommand{\descripthook}{\setlength{\topsep}{0pt}
  \setlength{\itemsep}{0pt}}
  
\usepackage{color}
\usepackage{graphicx}
\usepackage{amssymb,amsmath,amsfonts}
\usepackage{epsfig,longtable}
\usepackage[tight,hang]{subfigure}
\usepackage{algorithm}
\usepackage{algorithmic}

\newcommand{\ignore}[1]{}
\renewcommand{\vec}[1]{\mathbf{#1}}

\newcommand{\norm}[1]{\Vert #1\Vert }

\newcommand{\mZ}{\mathbb{Z}}

\usepackage{caption}

\usepackage[hyperpageref]{backref}

\title{Slide Reduction, Revisited---\\Filling the Gaps in SVP Approximation} 
	\author{
Divesh Aggarwal\thanks{National University of Singapore. E-mail: \texttt{dcsdiva@nus.edu.sg}. 
This research was partially funded by the Singapore Ministry of Education and the National Research Foundation under grant R-710-000-012-135} 
\and Jianwei Li\thanks{National University of Singapore. E-mail: \texttt{lijianweithu@sina.com}.}
\and Phong Q. Nguyen\thanks{Inria and CNRS, JFLI,University of Tokyo. E-mail:\texttt{pnguyen@inria.fr}.}
\and Noah Stephens-Davidowitz\thanks{Massachusetts Institute of Technology. E-mail: \texttt{noahsd@gmail.com}. Supported by Vinod Vaikuntanathan's NSF-BSF grant number 1718161 and NSF CAREER Award number 1350619.}
}

\date{}

\begin{document}

\maketitle

\begin{abstract}
    We show how to generalize Gama and Nguyen's slide reduction algorithm [STOC '08] for solving the approximate Shortest Vector Problem over lattices (SVP). As a result, we show the fastest provably correct algorithm for $\delta$-approximate SVP for all approximation factors $n^{1/2+\eps} \leq \delta \leq n^{O(1)}$. This is the range of approximation factors most relevant for cryptography.
\end{abstract}

	\section{Introduction}
	
	A lattice $\lat \subset \R^m$ is the set of integer linear combinations  
	\[
	\lat := \lat(\basis) = \{z_1 \vec{b}_1 + \cdots + z_n \vec{b}_n \ : \ z_i \in \Z \}
	\]
	of linearly independent basis vectors $\basis = (\vec{b}_1,\ldots, \vec{b}_n) \in \R^{m \times n}$. 
	We call $n$ the \emph{rank} of the lattice.
	
	The Shortest Vector Problem (SVP) is the computational search problem in which the input is (a basis for) a lattice $\lat \subseteq \Z^m$, and the goal is to output a non-zero lattice vector $\vec{y} \in \lat$ with minimal length, $\|\vec{y}\| = \lambda_1(\lat) := \min_{\vec{x}\in \lat_{\neq \vec0}} \|\vec{x}\|$. For $\delta \geq 1$, the $\delta$-approximate variant of SVP ($\delta$-SVP) is the relaxation of this problem in which any non-zero lattice vector $\vec{y}\in \lat_{\neq \vec0}$ with $\|\vec{y}\| \leq \delta \cdot \lambda_1(\lat)$ is a valid solution. 
	
	A closely related problem is \emph{$\delta$-Hermite SVP} ($\delta$-HSVP, sometimes also called Minkowski SVP), which asks us to find a non-zero lattice vector $\vec{y}\in \lat_{\neq \vec0}$ with $\|\vec{y}\| \leq \delta \cdot \mathrm{vol}(\lat)^{1/n}$, where $\mathrm{vol}(\lat) := \det(\basis^T \basis)^{1/2}$ is the covolume of the lattice. \emph{Hermite's constant} $\gamma_n$ is (the square of) the minimal possible approximation factor that can be achieved in the worst case. I.e., 
	\[
	    \gamma_n := \sup \frac{\lambda_1(\lat)^{2}}{\vol(\lat)^{2/n}}
	    \; ,
	\]
	where the supremum is over lattices $\lat \subset \R^n$ with full rank $n$.
	Hermite's constant is only known exactly for $1 \leq n \leq 8$ and $n = 24$, but it is known to be asymptotically linear in $n$, i.e., $\gamma_n = \Theta(n)$. HSVP and Hermite's constant play a large role in algorithms for $\delta$-SVP.
	
	Starting with the celebrated work of Lenstra, Lenstra, and Lov{\'a}sz in 1982~\cite{LLLFactoringPolynomials82}, algorithms for solving $\delta$-(H)SVP for a wide range of parameters $\delta$ have found innumerable applications, including
	factoring polynomials over the rationals~\cite{LLLFactoringPolynomials82}, integer programming~\cite{LenIntegerProgramming83,KanImprovedAlgorithms83,DPVEnumerativeLattice11}, cryptanalysis~\cite{ShaPolynomialtimeAlgorithm84,OdlRiseFall90,JSLatticeReduction98,NSTwoFaces01}, etc. More recently, many cryptographic primitives have been constructed whose security is based on the (worst-case) hardness of $\delta$-SVP or closely related lattice problems \cite{AjtGeneratingHard96,RegLatticesLearning09,GPVTrapdoorsHard08,PeiPublickeyCryptosystems09,PeiDecadeLattice16}. Such lattice-based cryptographic constructions are likely to be used on massive scales (e.g., as part of the TLS protocol) in the not-too-distant future~\cite{NISPostQuantumCryptography18}, and in practice, the security of these constructions depends on the fastest algorithms for $\delta$-(H)SVP, typically for $\delta = \poly(n)$.
	
	Work on $\delta$-(H)SVP has followed two distinct tracks. There has been a long line of work showing progressively faster algorithms for exact SVP (i.e., $\delta = 1$)~\cite{KanImprovedAlgorithms83,AKSSieveAlgorithm01,NVSieveAlgorithms08,PSSolvingShortest09,MVDeterministicSingle13}. However, even the fastest such algorithm (with proven correctness) runs in time $2^{n + o(n)}$~\cite{ADRSSolvingShortest15,ASJustTake18}. So, these algorithms are only useful for rather small $n$.
	
	This paper is part of a separate line of work on \emph{basis reduction algorithms}~\cite{LLLFactoringPolynomials82,SchHierarchyPolynomial87,SELatticeBasis94,GHKNRankinConstant06,GNFindingShort08,HPSAnalyzingBlockwise11,MWPracticalPredictable16}. (See~\cite{NVLLLAlgorithm10} and ~\cite{MWPracticalPredictable16} for a much more complete list of works on basis reduction.) At a high level, these are reductions from $\delta$-(H)SVP on lattices with rank $n$ to exact SVP on lattices with rank $k \leq n$. More specifically, these algorithms divide a basis $\basis$ into projected blocks $\basis_{[i,i+k-1]}$
	with \emph{block size} $k$, where $\basis_{[i,j]}=(\pi_{i}(\mathbf{b}_{i}),
\pi_{i}(\mathbf{b}_{i+1}), \ldots, \pi_{i}(\mathbf{b}_{j}))$ and  $\pi_i$ is the orthogonal projection  onto
the subspace orthogonal to $\vec{b}_1,\ldots, \vec{b}_{i-1}$.
	Basis reduction algorithms use their SVP oracle to find short vectors in these (low-rank) blocks and incorporate these short vectors into the lattice basis $\basis$. By doing this repeatedly (at most $\poly(n, \log \|\basis\|)$ times) with a cleverly chosen sequence of blocks, such algorithms progressively improve the ``quality'' of the basis $\basis$ until $\vec{b}_1$ is a solution to $\delta$-(H)SVP for some $\delta \geq 1$. The goal, of course, is to take the block size $k$ to be small enough that we can actually run an exact algorithm on lattices with rank $k$ in reasonable time while still achieving a relatively good approximation factor $\delta$.
	
	For HSVP, the DBKZ algorithm due to Micciancio and Walter yields the best proven approximation factor for all ranks $n$ and block sizes $k$~\cite{MWPracticalPredictable16}. Specifically, it achieves an approximation factor of 
	\begin{equation}
	\label{eq:MW_intro}
	\delta_{\mathsf{MW},H} := \gamma_k^{\frac{n-1}{2(k-1)}}
	\; .
	\end{equation}
	(Recall that $\gamma_k = \Theta(k)$ is Hermite's constant. Here and throughout the introduction, we have left out low-order factors that can be made arbitrarily close to one.)
	Using a result due to Lov\'asz~\cite{LovAlgorithmicTheory86}, this can be converted into an algorithm for $\delta_{\mathsf{MW},H}^2$-SVP.
	However, the slide reduction algorithm of Gama and Nguyen~\cite{GNFindingShort08} achieves a better approximation factor for SVP. It yields
	\begin{equation}
	   \label{eq:slide_approx_ceil}
	\delta_{\mathsf{GN},H} := \gamma_k^{\frac{\ceil{n}_k-1}{2(k-1)}}
	\qquad 
	\delta_{\mathsf{GN},S} := 
	\gamma_k^{\frac{\ceil{n}_k-k}{k-1}}
	\; ,
	\end{equation}
	for HSVP and SVP respectively, where we write $\ceil{n}_k := k \cdot \ceil{n/k}$ for $n$ rounded up to the nearest multiple of $k$. (We have included the result for HSVP in Eq.~\eqref{eq:slide_approx_ceil} for completeness, though it is clearly no better than Eq.~\eqref{eq:MW_intro}.)

	The discontinuous approximation factor in Eq.~\eqref{eq:slide_approx_ceil} is the result of an unfortunate limitation of slide reduction: it only works when the block size $k$ divides the rank $n$. If $n$ is not divisible by $k$, then we must artificially pad our basis so that it has rank $\ceil{n}_k$, which results in the rather odd expressions in Eq.~\eqref{eq:slide_approx_ceil}. Of course, for $n \gg k$, this rounding has little effect on the approximation factor. But, for cryptographic applications, we are interested in small polynomial approximation factors $\delta \approx n^c$ for relatively small constants $c$, i.e., in the case when $k = \Theta(n)$. For such values of $k$ and $n$, this rounding operation can cost us a constant factor in the exponent of the approximation factor, essentially changing $n^c$ to $n^{\ceil{c}}$. Such constants in the exponent have a large effect on the security of lattice-based cryptography.\footnote{The security of lattice-based cryptography is actually assessed using heuristic algorithms that outperform Eq.~\eqref{eq:slide_approx_ceil} empirically~\cite{APSConcreteHardness15}, so that Eq.~\eqref{eq:slide_approx_ceil} is not and should not be used directly for this purpose. In this work, we restrict our attention to what we can prove.}

\subsection{Our results}

Our first main contribution is a generalization of Gama and Nguyen's slide reduction~\cite{GNFindingShort08} without the limitation that the rank $n$ must be a multiple of the block size $k$. Indeed, we achieve exactly the approximation factor shown in Eq.~\eqref{eq:slide_approx_ceil} without any rounding, as we show below.

As a very small additional contribution, we allow for the possibility that the underlying SVP algorithm for lattices with rank $k$
only solves $\delta$-approximate SVP for some $\delta > 1$.
This technique was already known to folklore and used in practice, and the proof requires no new ideas. Nevertheless, we believe that this work is the first to formally show that a $\delta$-SVP algorithm suffices and
to compute the exact dependence on $\delta$.
(This minor change proves quite useful when we instantiate our $\delta$-SVP subroutine with the $2^{0.802k}$-time $\delta$-SVP algorithm for some large constant $\delta \gg 1$  due to Liu, Wang, Xu, and Zheng~\cite{LWXZShortestLattice11,WLWFindingShortest15}. See Table~\ref{Table:approximat SVP} and Figure~\ref{fig:running_time}.)

\begin{theorem}[Informal, slide reduction for $n \geq 2k$]\label{thm:main_result_intro}
	For any approximation factor $\delta \geq 1$ and block size $k := k(n) \geq 2$, there is an efficient reduction from $\delta_H$-HSVP and $\delta_S$-SVP on lattices with rank $n \geq 2k$  
	to $\delta$-SVP on lattices with rank $k$, 
	where
	\[
	\delta_H := (\delta^2 \gamma_{k})^{\frac{n-1}{2(k-1)}}
	\qquad
	\delta_S :=
	\delta(\delta^2 \gamma_{k})^{\frac{n-k}{k-1}}
	\; .
	\]
\end{theorem}
Notice in particular that this matches Eq.~\eqref{eq:slide_approx_ceil} in the case when $\delta = 1$ and $k$ divides $n$. (This is not surprising, since our algorithm is essentially identical to the original algorithm from~\cite{GNFindingShort08} in this case.) Theorem~\ref{thm:main_result_intro} also matches the approximation factor for HSVP achieved by~\cite{MWPracticalPredictable16}, as shown in Eq.~\eqref{eq:MW_intro}, so that the best (proven) approximation factor for both problems is now achieved by a single algorithm.

However, Theorem~\ref{thm:main_result_intro} only applies for $n \geq 2k$. 
Our second main contribution is an algorithm that works for $k \leq n \leq 2k$. To our knowledge, this is the first algorithm that provably achieves sublinear approximation factors
for SVP and is asymptotically faster than, say, the fastest algorithm for $O(1)$-SVP. (We overcame a small barrier here. See the discussion in Section~\ref{sec:bislide}.)

\begin{theorem}[Informal, slide reduction for $n \leq 2k$]
	\label{thm:bi-slide_intro}
	For any approximation factor $\delta \geq 1$ and block size $k \in [n/2,n]$, there is an efficient reduction from $\delta_S$-SVP on lattices with rank 
	$n$ to $\delta$-SVP on lattices with rank $k$, where 
	\[
	    \delta_S := \delta^2 \sqrt{\gamma_k}  (\delta^2 \gamma_{q})^{\frac{q+1}{q-1} \cdot \frac{n-k}{2k}} 
	    \lesssim \delta (\delta^2 \gamma_k)^{\frac{n}{2k}}
	    \; ,
	\]
	and $q := n-k \leq k$.
\end{theorem}

Together, these algorithms yield the asymptotically fastest proven 
running times 
for $\delta$-SVP for all approximation factors $n^{1/2 + \eps} \leq \delta \leq n^{O(1)}$---with a particularly large improvement when $\delta = n^c$ for $1/2 < c < 1$ or for any $c$ slightly smaller than an integer. Table \ref{Table:approximat SVP} and Figure~\ref{fig:running_time} summarize the current state of the art.

\begin{table}[t] 
\begin{center}
	\scalebox{0.88}{
	\begin{tabular}{|l|c l| c l | c l|}
		\hline
		Approximation factor & \multicolumn{2}{|l|}{Previous best}& \multicolumn{2}{|l|}{Folklore} & \multicolumn{2}{|l|}{This work} \\ \hline
		Exact  &$2^n$ &\cite{ADRSSolvingShortest15} &---& &  --- &\\ 
		$\Omega(1) \leq \delta \leq \sqrt{n}$ &$2^{0.802n}$ &\cite{WLWFindingShortest15} &---& & --- &\\ 
		$n^c$ for 
		$c \in [\frac{1}{2},1]$& $2^{0.802n}$  &\cite{WLWFindingShortest15} &--- & & $2^{\frac{0.802n}{2c}}$ &[*]$+$\cite{WLWFindingShortest15}\\ 
		$n^c$ for  $c\ge 1$ 
		&$2^{\frac{n}{\lfloor c+1\rfloor}}$  &\cite{GNFindingShort08}+\cite{ADRSSolvingShortest15}& $2^{ \frac{0.802 n}{\lfloor c+1\rfloor}}$ &\cite{GNFindingShort08}+\cite{WLWFindingShortest15} &
		$2^{\frac{0.802n}{c+1}}$  &[*]$+$\cite{WLWFindingShortest15}\\ \hline
	\end{tabular}}
		\caption{\label{Table:approximat SVP}Algorithms for solving SVP. We write [A]$+$[B] to denote the algorithm that uses basis reduction from [A] with the exact/near-exact SVP algorithm from [B], and we write [*] for this work. The ``folklore'' column represents a result that was likely known to many experts in the field but apparently never published.
	}
\end{center}
\end{table}

\begin{figure}[t]
	\begin{center}
	    \includegraphics[width=0.7\textwidth]{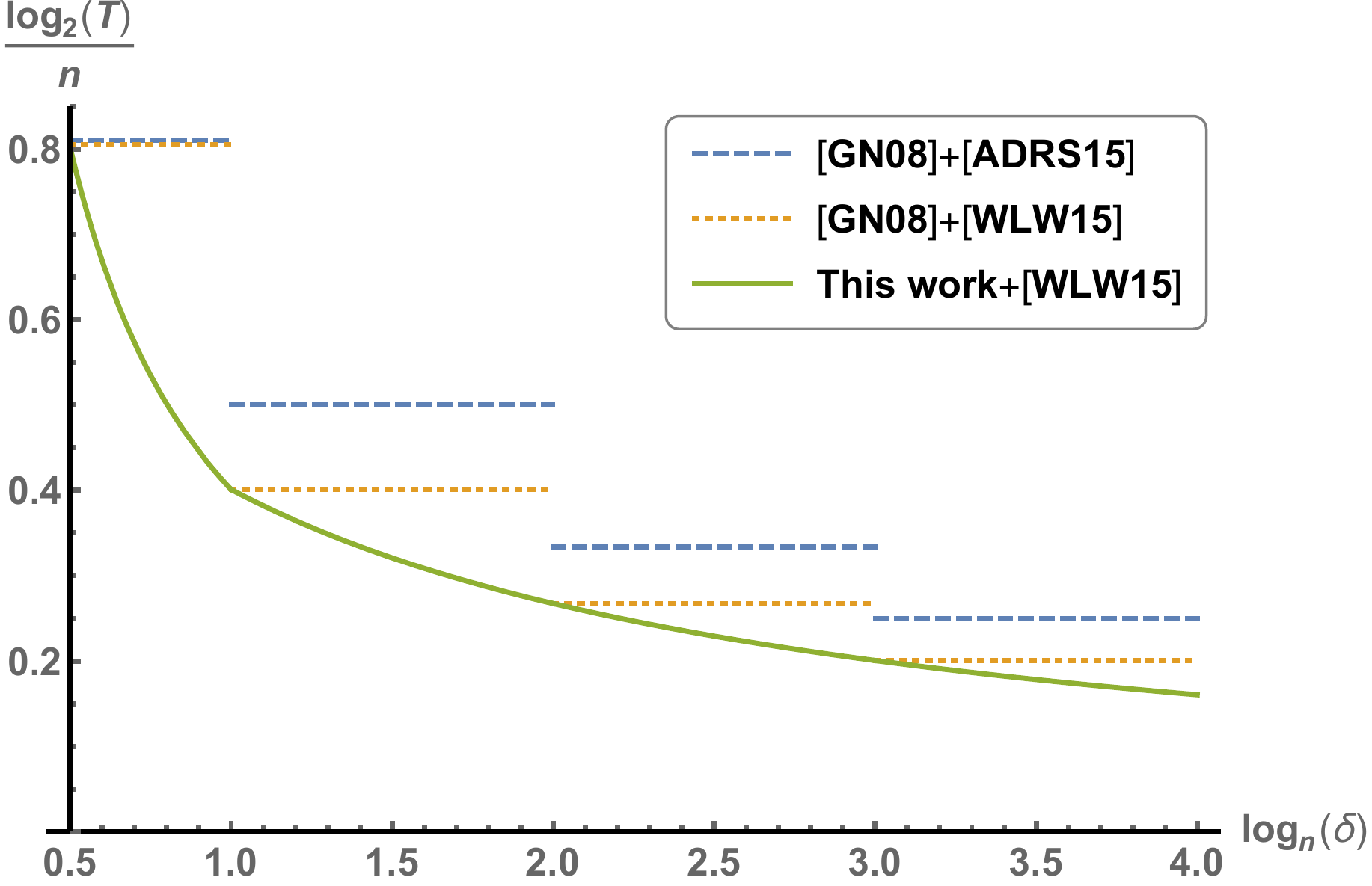}
	    \caption{\label{fig:running_time}Running time $T$ as a function of approximation factor $\delta$ for $\delta$-SVP. The $y$-axis is $\log_2(T)/n$, and the $x$-axis is $\log_n \delta$.}
	\end{center}
\end{figure}

\subsection{Our techniques}
\label{sec:techniques}

We first briefly recall some of the details of Gama and Nguyen's slide reduction. Slide reduction divides the basis $\basis = (\vec{b}_1,\ldots, \vec{b}_n) \in \R^{m \times n}$ evenly into disjoint ``primal blocks'' $\basis_{[ik+1,(i+1)k]}$ of length $k$. (Notice that this already requires $n$ to be divisible by $k$.) It also defines certain ``dual blocks'' $\basis_{[ik+2,(i+1)k+1]}$, which are the primal blocks shifted one to the right. The algorithm then tries to simultaneously satisfy certain primal and dual conditions on these blocks. Namely, it tries to \emph{SVP-reduce} each primal block---i.e., it tries to make the first vector in the block $\vec{b}_{ik+1}^*$ a shortest vector in $\lat(\basis_{[ik+1,(i+1)k]})$, where $\vec{b}_{j}^* := \pi_j(\vec{b}_j)$. Simultaneously, it tries to \emph{dual SVP-reduce} (DSVP-reduce) the dual blocks. (See Section~\ref{subsec:lattice reduction} for the definition of DSVP reduction.)
We call a basis that satisfies all of these conditions simultaneously \emph{slide-reduced}.

\begin{figure}[t]
\begin{center}
\begin{boxedminipage}{0.95 \textwidth}
\begin{center}
\includegraphics[scale=0.75]{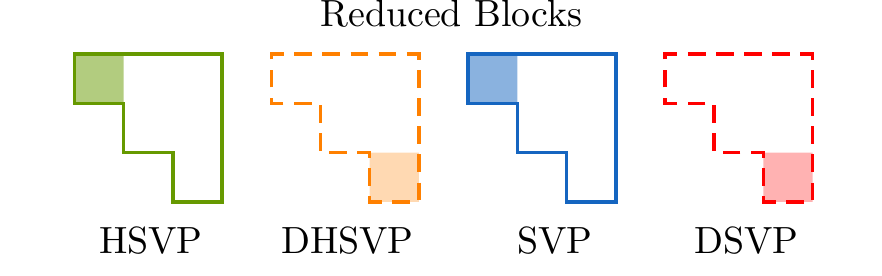}\\
\includegraphics[width = 0.45\textwidth]{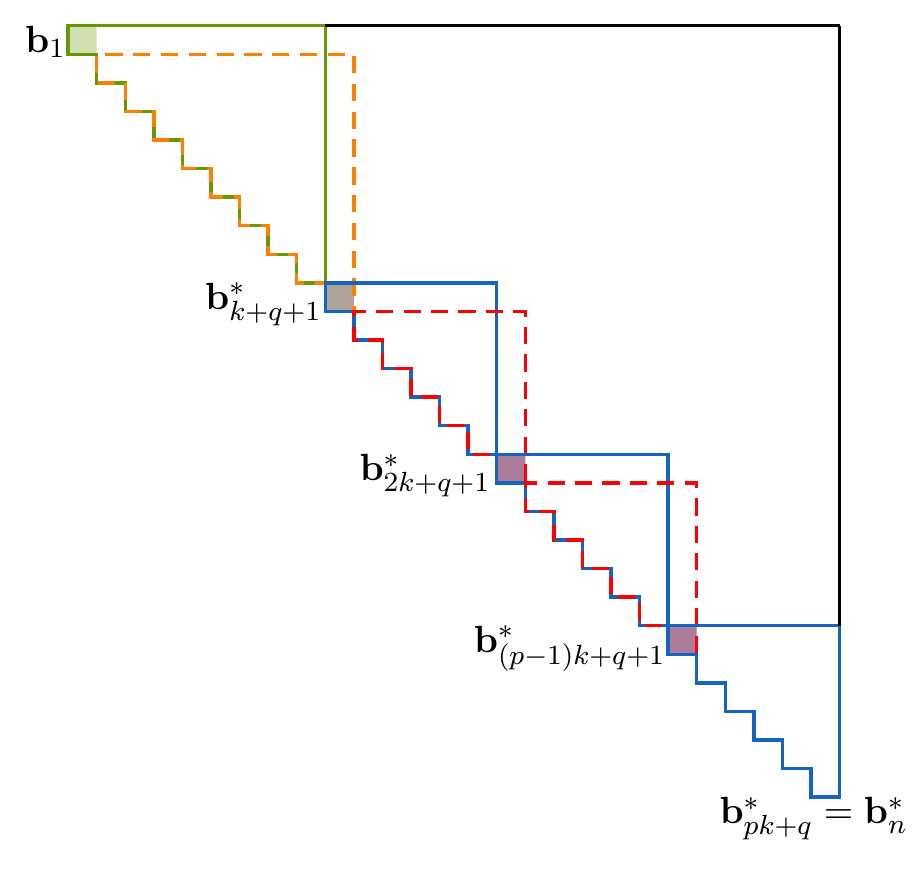} \quad \includegraphics[width = 0.45\textwidth]{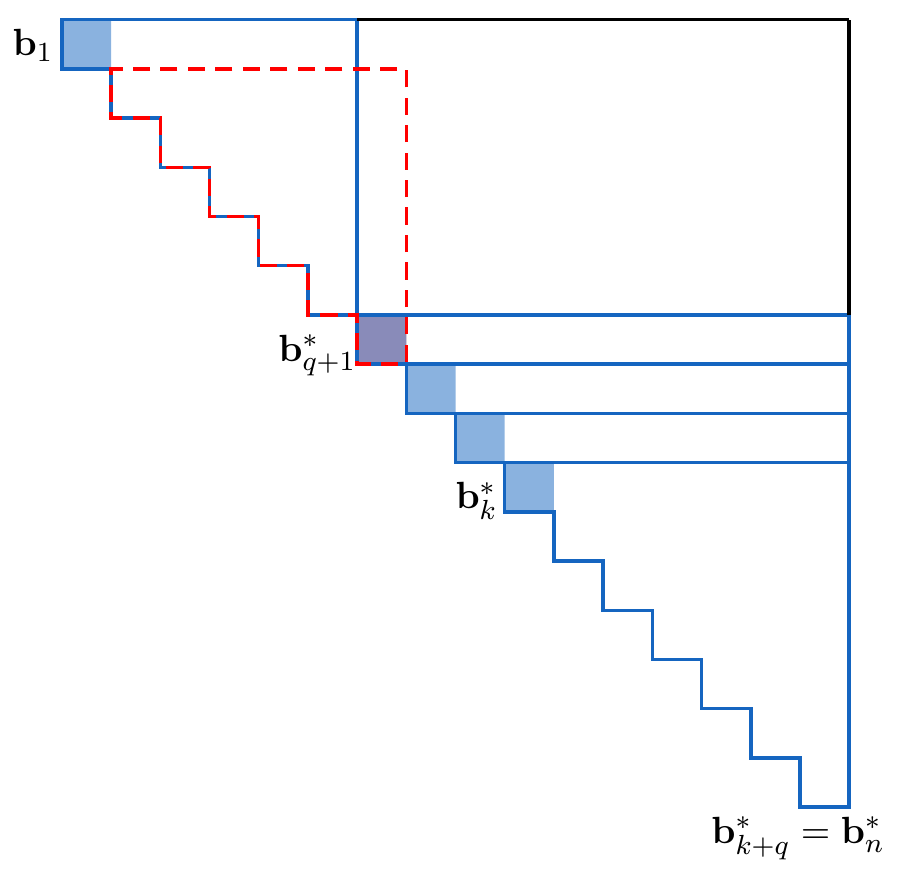}
\end{center}
\end{boxedminipage}
\caption{ \label{fig:SlideReduction}Slide reduction of an upper-triangular matrix for $n = pk + q \geq 2k$ (left) and $n = k + q \leq 2k$ (right). (The original notion of slide reduction in~\cite{GNFindingShort08} used only SVP-reduced and DSVP-reduced blocks of fixed size $k$.)}
\end{center}
\end{figure}

An SVP oracle for lattices with rank $k$ is sufficient to enforce all primal conditions or all dual conditions separately. (E.g., we can enforce the primal conditions by simply finding a shortest non-zero vector in each primal block and including this vector in an updated basis for the block.)
Furthermore, if all primal and dual conditions hold simultaneously, then $\|\vec{b}_1\| \leq \delta_{\mathsf{GN},S} \lambda_1(\lat)$ with $\delta_{\mathsf{GN},S}$ as in Eq.~\eqref{eq:slide_approx_ceil}, so that $\|\vec{b}_1\|$ yields a solution to $\delta_{\mathsf{GN},S}$-SVP. This follows from repeated application of a ``gluing'' lemma on such bases, which shows how to ``glue together'' two reduced block to obtain a larger reduced block. (See Lemma~\ref{lem:pilingup}.) 
Finally, Gama and Nguyen showed that, if we alternate between SVP-reducing the primal blocks and DSVP-reducing the dual blocks, then the basis will converge quite rapidly to a slide-reduced basis (up to some small slack)~\cite{GNFindingShort08}. Combining all of these facts together yields the main result in~\cite{GNFindingShort08}. (See Section~\ref{sec:Generalized-slide-reduction}.)

\paragraph{The case $n > 2k$.} We wish to extend slide reduction to the case when $n = p k + q$ for $1 \leq q < k$. So, intuitively, we have to decide what to do with ``the extra $q$ vectors in the basis.''

We start by observing that the analysis of slide reduction (and, in particular, this ``gluing'' property) does not quite require the first block $\basis_{[1,k]}$ to be SVP-reduced. Instead, it essentially only requires it to be ``HSVP-reduced.'' I.e., we do not really need $\|\vec{b}_1\| \leq \delta_S \lambda_1(\lat(\basis_{[1,k]}))$; we basically only need $\|\vec{b}_1\| \leq \delta_H \vol(\basis_{[1,k]})^{1/k}$. Something similar holds for the first dual block, so that at least for the first block and the corresponding dual block, we basically only need an HSVP oracle.\footnote{\label{foot:degenerate}We are ignoring a certain degenerate case here for simplicity. Namely, if all short vectors happen to lie in the span of the first block, and these vectors happen to be very short relative to the volume of the first block, then calling an HSVP oracle on the first block might not be sufficient to solve approximate SVP. Of course, if we know a low-dimensional subspace that contains the shortest non-zero vector, then finding short lattice vectors is much easier. This degenerate case is therefore easily handled separately (but it does in fact need to be handled separately).}

This suggests that we might want to simply add the extra $q$ vectors to the first block. I.e., we can take one ``big block'' $\basis_{[1,k+q]}$ of length $k + q$, and $p-1$ ``regular'' blocks $\basis_{[ik+q+1,(i+1)k + q]}$ of length $k$. The regular blocks satisfy the same conditions as in~\cite{GNFindingShort08}---they are SVP-reduced and the corresponding dual blocks are DSVP-reduced. For the big first block, we replace SVP reduction
by an appropriate notion of HSVP reduction. Similarly, we replace DSVP reduction of the (big) first dual block by the appropriate dual notion of HSVP reduction.

To get the best results, we instantiate our HSVP oracle with the algorithm from~\cite{MWPracticalPredictable16}. Since we only need an oracle for HSVP, we are able to take advantage of the very impressive approximation factor achieved by~\cite{MWPracticalPredictable16} for this problem (i.e., Eq.~\eqref{eq:MW_intro}). In fact, the approximation factor achieved by~\cite{MWPracticalPredictable16} is exactly what we need to apply our gluing lemma. (This is not a coincidence, as we explain in Section~\ref{sec:Generalized-slide-reduction}.) The result is Theorem~\ref{thm:main_result_intro}.

\paragraph{The case $n < 2k$. } For $n = k + q < 2k$, the above idea cannot work. In particular, a ``big block'' of size $k + q$ in this case would be our entire basis! So, instead of working with one big block and some ``regular blocks'' of size $k$, we work with a ``small block'' of size $q$ and one regular block of size $k$. We then simply perform slide reduction with (primal) blocks $\basis_{[1,q]}$ and $\basis_{[q+1,n]} = \basis_{[n-k+1,n]}$. If we were to stop here, we would achieve an approximation factor of roughly $\gamma_q$, which for $q = \Theta(k)$ is essentially the same as the approximation factor of roughly $\gamma_k$ that we get when the rank is $2k$. I.e., we would essentially ``pay for two blocks of length $k$,'' even though one block has size $q < k$.

However, we notice that a slide-reduced basis guarantees more than just a short first vector. It also promises a very strong bound on $\vol(\basis_{[1,q]})$. In particular, since $q < k$ and since we have access to an oracle for lattices with rank $k$, it is natural to try to extend this small block $\basis_{[1,q]}$ with low volume to a larger block $\basis_{[1,k]}$ of length $k$ that still has low volume. Indeed, we can use our SVP oracle to guarantee that $\basis_{[q+1,k]}$ consists of relatively short vectors so that $\vol(\basis_{[q+1,k]})$ is relatively small as well. (Formally, we SVP-reduce $\basis_{[i,n]}$ for $i \in [q+1,k]$. Again, we are ignoring a certain degenerate case, as in Footnote~\ref{foot:degenerate}.) This allows us to upper bound $\vol(\basis_{[1,k]}) = \vol(\basis_{[1,q]}) \cdot \vol(\basis_{[q+1,k]})$, which implies that  $\lambda_1(\lat(\basis_{[1,k]}))$ is relatively short. We can therefore find a short vector by making an additional SVP oracle call on $\lat(\basis_{[1,k]})$. (Micciancio and Walter used a similar idea in~\cite{MWPracticalPredictable16}.)

\subsection{Open questions and directions for future work}

Table~\ref{Table:approximat SVP} suggests an obvious open question: can we find a non-trivial basis reduction algorithm that provably solves $\delta$-SVP for $\delta \leq O(\sqrt{n})$?
 More formally, can we reduce $O(\sqrt{n})$-SVP on lattices with rank $n$ to exact SVP on lattices with rank $k = c n$ for some constant $c < 1$. Our current proof techniques seem to run into a fundamental barrier here in that they seem more-or-less incapable of achieving $\delta \ll \sqrt{\gamma_{k}}$. This setting is interesting in practice, as many record lattice computations use block reduction with $k \ge n/2$ as a subroutine, such as~\cite{CNFasterAlgorithms12}. (One can provably achieve approximation factors $\delta \ll \sqrt{\gamma_k}$ when $k = (1-o(1))n$ with a bit of work,\footnote{For example, it is immediate from the proof of Theorem~\ref{th:property of bi-slide-reduction} that the (very simple) notion of a slide-reduced basis for $n \leq 2k$ in Definition~\ref{def:slide-reduction} is already enough to obtain $\delta \approx \gamma_{n-k} \approx n- k$. So, for $n \lesssim k + \sqrt{k}$, this already achieves $\delta \lesssim \sqrt{n}$. With a bit more work, one can show that an extra oracle call like the one used in Corollary~\ref{cor:main1} can yield a still better approximation factor in this rather extreme setting of $k = (1-o(1))n$.} but it is not clear if these extreme parameters are useful.)

Next, we recall that this work shows how to exploit the existing very impressive algorithms for HSVP (in particular, DBKZ~\cite{MWPracticalPredictable16}) to obtain better algorithms for SVP. This suggests two closely related questions for future work: (1) can we find better algorithms for HSVP (e.g., for $\delta$-HSVP with $\delta \approx \sqrt{\gamma_n}$---i.e., ``near-exact'' HSVP); and (2) where else can we profitably replace SVP oracles with HSVP oracles? Indeed, most of our analysis (and the analysis of other basis reduction algorithms) treats the $\delta$-SVP oracle as a $\delta\sqrt{\gamma_k}$-HSVP oracle. We identified one way to exploit this to actually get a faster algorithm, but perhaps more can be done here---particularly if we find faster algorithms for HSVP.

We also leave it to future work to implement our algorithms and to study how they perform in practice. Indeed, Micciancio and Walter showed that (a slightly optimized version of) slide reduction is competitive with even the best heuristic algorithms in practice, in terms of both the running time  and the approximation factor~\cite{MWPracticalPredictable16}. Since our algorithms are generalizations of slide reduction, one might guess that they also perform well in practice. We leave it to others to confirm or refute this guess.

Finally, we note that we present two distinct (though similar) algorithms: one for lattices with rank $n \leq 2k$ and one for lattices with rank $n \geq 2k$. It is natural to ask whether there is a single algorithm that works in both regimes. Perhaps work on this question could even lead to better approximation factors.

\section{Preliminaries}

\label{sec:back}
We denote column vectors $\vec{x} \in \R^m$ by bold lower-case letters. Matrices $\basis \in \R^{m \times n}$ are denoted by bold upper-case letters, and we often think of a matrix as a list of column vectors, $\basis = (\vec{b}_1,\ldots, \vec{b}_n)$.
For a matrix $\vec{B} =(\mathbf{b}_{1}, \ldots, \mathbf{b}_{n})$ with $n$ linearly independent columns,
we write $\lat(\vec{B}) :=\{z_1 \vec{b}_1 + \cdots + z_n \vec{b}_n \ : \ z_i \in \Z \}$ for the lattice
generated by $\basis$ and $\|\vec{B}\| = \max \{\|\mathbf{b}_{1}\|,\ldots, \|\mathbf{b}_{n}\|\}$ for 
the maximum norm of a column.
We often implicitly assume that $m \geq n$ and that a basis matrix $\basis \in \R^{m \times n}$ has rank $n$ (i.e., that the columns of $\basis$ are linearly independent).
We use the notation $\log := \log_2 $ to mean the logarithm with base two.

\subsection{Lattices}\label{subsec:GSO} 

For any lattice $\lat$, its {\em dual lattice}
is \[
\lat^\times = \{\mathbf{w} \in \spn(\lat):\ \langle \mathbf{w}, \mathbf{y}\rangle \in
\mathbb{Z}\ \textrm{for\ all}\ \vec{y} \in \lat \}
\; .
\]
If $\basis\in \R^{m \times n}$ is a basis of $\lat$, then $\lat^\times$ has basis $\basis^\times := \basis(\basis^T \basis)^{-1}$, 
called
the {\em dual basis} of $\basis$.
The {\em reversed dual basis} $\basis^{-s}$ of $\basis$ is simply $\basis^\times$ with its columns in reversed order~\cite{GHNSymplecticLattice06}.

\subsection{Gram-Schmidt-Orthogonalization}
For 
a basis $\basis = (\vec{b}_1,\ldots, \vec{b}_n) \in \R^{m \times n}$, we associate a sequence of projections $\pi_{i} := \pi_{\{\vec{b}_1,\ldots, \vec{b}_{i-1}\}^\perp}$. Here, $\pi_{W^\perp}$ means the orthogonal 
projection onto the subspace $W^\perp$ orthogonal to $W$.
As in~\cite{GNFindingShort08}, $\basis_{[i,j]}$ denotes the projected block $(\pi_{i}(\vec{b}_i),\pi_{i}(\vec{b}_{i+1}),\ldots, \pi_{i}(\vec{b}_j))$.

We also associate to 
$\basis$ its Gram-Schmidt orthogonalization (GSO) $\vec{B}^{\ast} := (\mathbf{b}_1^{\ast}, \ldots, \mathbf{b}_{n}^{\ast})$, where $\vec{b}_i^* := \pi_{i}(\vec{b}_i) =  \vec{b}_i - \sum_{j < i} \mu_{i,j} \vec{b}_j^*$, and $\mu_{i,j} = \langle \vec{b}_i, \vec{b}_j^* \rangle/\|\vec{b}_j^*\|^2$. 

We say that 
$\basis$ is \emph{size-reduced} if $|\mu_{i,j}| \le \frac{1}{2}$ for all $i \neq j$: then 
$\|\basis\| \le \sqrt{n}\|\basis^{\ast}\|$. Transforming a basis into this form without modifying $\lat(\basis)$ or $\basis^{\ast}$ 
is called {\em size reduction}, and this can be done easily and efficiently. 
\subsection{Lattice basis reduction}
\label{subsec:lattice reduction}

\paragraph{LLL reduction.} Let $\vec{B} = (\mathbf{b}_1,
\ldots, \mathbf{b}_n)$ be a size-reduced basis.
For $\eps \in [0, 1]$, we say that $\vec{B}$ is $\eps$-{\em LLL-reduced} \cite{LLLFactoringPolynomials82} if every rank-two 
projected  
block $\vec{B}_{[i,i+1]}$
satisfies Lov\'{a}sz's condition:
$\|\mathbf{b}_{i}^{\ast}\|^2 \leq (1 +
\varepsilon)\|\mu_{i,i-1}\mathbf{b}_{i-1}^{\ast} + \mathbf{b}_{i}^{\ast}\|^2$ for $1 < i \leq n$.
For $\eps \geq 1/\poly(n)$, one can efficiently compute an $\eps$-LLL-reduced basis for a given lattice.

\paragraph{SVP reduction and its extensions.}
Let $\vec{B} = (\mathbf{b}_1, \ldots, \mathbf{b}_{n})$ be a basis of a lattice $\lat$ 
and $\delta \ge 1$ be an approximation factor. 

We say that $\vec{B}$ is $\delta$-{\em SVP-reduced} if 
$\|\mathbf{b}_1\|  \leq \delta\cdot \lambda_{1}(\lat)$.
Similarly, we say that $\vec{B}$ is $\delta$-{\em HSVP-reduced}
if $\|\vec{b}_1\| \leq \delta \cdot \mathrm{vol}(\lat)^{1/n}$.

$\vec{B}$ is {\em $\delta$-DSVP-reduced} \cite{GNFindingShort08} (where D
stands for dual) if the reversed dual basis $\vec{B}^{-s}$ is $\delta$-SVP-reduced and $\basis$ is $\frac{1}{3}$-LLL-reduced. 
Similarly, we say that $\vec{B}$ is {\em $\delta$-DHSVP-reduced}
if $\vec{B}^{-s}$ is $\delta$-HSVP-reduced.

The existence of such $\delta$-DSVP-reduced bases is guaranteed by a classical property of LLL that $\|\vec{b}_{n}^{\ast}\|$ never decreases during the LLL-reduction process \cite{LLLFactoringPolynomials82}.

We can efficiently compute a 
$\delta$-(D)SVP-reduced
basis for a given rank $n$ lattice $\lat \subseteq \Z^m$ with access to an oracle for $\delta$-SVP on lattices with rank at most $n$. 
Furthermore, given a basis $\basis = (\vec{b}_1,\ldots, \vec{b}_n) \in \Z^{m\times n}$ of $\lat$ and an index $i \in [1,n-k+1]$, we can use a $\delta$-SVP oracle for lattices with rank at most $k$ 
to efficiently compute a size-reduced basis $\mathbf{C} = (\vec{b}_1,\ldots, \vec{b}_{i-1}, \vec{c}_i,\ldots, \vec{c}_{i+k-1}, \vec{b}_{i+k},\ldots, \vec{b}_n)$ of $\lat$ 
such that the block $\mathbf{C}_{[i,i+k-1]}$ is $\delta$-SVP reduced or $\delta$-DSVP reduced: 
\begin{itemize}
    \item If $\mathbf{C}_{[i,i+k-1]}$ is $\delta$-SVP-reduced, 
    the procedures in \cite{GNFindingShort08,MWPracticalPredictable16} equipped with $\delta$-SVP-oracle ensure that $\|\mathbf{C}^{\ast}\|\le \|\basis^{\ast}\|$;
    \item If $\mathbf{C}_{[i,i+k-1]}$ is $\delta$-DSVP-reduced, the inherent LLL reduction implies $\|\mathbf{C}^{\ast}\|\le 2^{k}\|\basis^{\ast}\|$. Indeed, the GSO 
    of $\mathbf{C}_{[i,i+k-1]}$ satisfies $\|(\mathbf{C}_{[i,i+k-1]})^{\ast}\|\le 2^{k/2}\lambda_{k}(\lat(\mathbf{C}_{[i,i+k-1]}))$ (by \cite[p.\,518, Line 27]{LLLFactoringPolynomials82}) and $\lambda_{k}(\lat(\mathbf{C}_{[i,i+k-1]}))\le \sqrt{k}\|\basis^{\ast}\|$. Here, $\lambda_k(\cdot)$ denotes the $k$-th minimum.
\end{itemize}

With size-reduction, we can iteratively perform $\poly(n, \log \|\basis\|)$ many such operations efficiently. In particular,
doing so will not increase $\|\basis^{\ast}\|$ by more than a factor of $2^{\poly(n,\log\|B\|)}$, and therefore the same is true of $\|\basis\|$. 
That is, all intermediate entries and the total cost during execution (excluding 
oracle queries) remain polynomially bounded in the initial input size; See, e.g., \cite{GNFindingShort08,LNApproximatingDensest14} for the evidence.
Therefore, to bound the running time of basis reduction, it suffices 
to bound the number of calls to these block reduction 
subprocedures.

\paragraph{Twin reduction and gluing.} We define the following notion, which was implicit in~\cite{GNFindingShort08} and will arise repeatedly in our proofs. 
$\vec{B} = (\vec{b}_1,\ldots, \vec{b}_{d+1})$  is $\delta$-{\em twin-reduced} if $\basis_{[1,d]}$ is $\delta$-{\em HSVP-reduced}
and $\basis_{[2,d+1]}$ is $\delta$-{\em DHSVP-reduced}.
The usefulness of twin reduction is illustrated by the following fact, which is the key idea behind Gama and Nguyen's slide reduction (and is remarkably simple in hindsight).

\begin{fact}
    \label{fact:twinsies}
    If $\basis := (\vec{b}_1,\ldots, \vec{b}_{d+1}) \in \R^{m \times (d+1)}$ is $\delta$-twin-reduced, then 
    \begin{equation}
        \label{eq:twin_decay}
        \|\vec{b}_1\| \le \delta^{2d/(d-1)} \|\vec{b}^*_{d+1}\|
        \; .
    \end{equation}
    Furthermore, 
    \begin{equation}
        \label{eq:twin_volume}
        \delta^{-d/(d-1)} \|\vec{b}_1\|   \le \vol(\basis)^{1/(d+1)} \le \delta^{d/(d-1)} \|\vec{b}^*_{d+1}\| 
        \; .
    \end{equation}
\end{fact}
\begin{proof}
    By definition, we have 
    $
        \|\vec{b}_1\|^d \leq \delta^d \vol(\basis_{[1,d]})
    $,
    which is equivalent to
    \[
    \|\vec{b}_1\|^{d-1} \leq \delta^d \vol(\basis_{[2,d]})
    \; .
    \]
    Similarly,
    \[
        \vol(\basis_{[2,d]}) \leq \delta^{d} \|\vec{b}_{d+1}^*\|^{d-1}
        \; .
    \]
    Combining these two inequalities yields Eq.~\eqref{eq:twin_decay}.
    
    Finally, we have
        $\|\vec{b}_1\|^d \|\vec{b}_{d+1}^*\| \leq \delta^{d} \vol(\basis)$.
    Applying Eq.~\eqref{eq:twin_decay} implies the first inequality in Eq.~\eqref{eq:twin_volume}, and similar analysis yields the second inequality.
\end{proof}

The following gluing lemma,
which is more-or-less implicit in prior work,
shows conditions on the blocks $\basis_{[1,d]}$ and $\basis_{[d+1,n]}$ that are sufficient to imply (H)SVP reduction of the full basis $\basis$. Notice in particular that the decay of the Gram-Schmidt vectors guaranteed by Eq.~\eqref{eq:twin_decay} is what is needed for Item~\ref{item:piling_HSVP} of the lemma below, when $\eta = \delta^{1/(d-1)}$. And, with this same choice of $\eta$, the HSVP reduction requirement on $\basis_{[1,d]}$ in Fact~\ref{fact:twinsies} is the same as the one in Item~\ref{item:piling_HSVP} of Lemma~\ref{lem:pilingup}.

\begin{lemma}[The gluing lemma] \label{lem:pilingup}
Let $\basis := (\vec{b}_1,\ldots, \vec{b}_n) \in \R^{m \times n}$, $\alpha, \beta, \eta \geq 1$, and $1 \le d \le n$.
\begin{enumerate}
    \item If $\basis_{[d+1,n]}$ is $\beta$-SVP-reduced,
$\|\vec{b}_1\| \le \alpha \|\vec{b}_{d+1}^*\|$, and
$\lambda_1(\lat(\basis)) < \lambda_1(\lat(\basis_{[1,d]}))$,
then $\basis$ is $\alpha \beta$-SVP-reduced. \label{item:piling_SVP}
    \item If $\basis_{[1,d]}$ is $\eta^{d-1}$-HSVP-reduced,
    $\basis_{[d+1,n]}$ is $\eta^{n-d-1}$-HSVP-reduced, and $\|\vec{b}_1\| \le \eta^{2d} \|\vec{b}_{d+1}^*\|$,
then $\basis$ is $\eta^{n-1}$-HSVP-reduced. \label{item:piling_HSVP}
\end{enumerate}
\end{lemma}
\begin{proof}
For Item~\ref{item:piling_SVP},
since $\lambda_1(\lat(\basis)) < \lambda_1(\lat(\basis_{[1,d]}))$,
there exists a shortest non-zero vector $\vec{u} \in \lat(\basis)$
with $\|\vec{u}\| = \lambda_1( \lat(\basis))$
and $\pi_d(\vec{u}) \ne 0$.
Since $\basis_{[d+1,n]}$ is $\beta$-SVP-reduced,
it follows that $\|\vec{b}_{d+1}^* \|/\beta \le  \|\pi_d(\vec{u})\|
\le \|\vec{u}\| = \lambda_1(\lat(\basis))$.
Finally, we have 
$
\|\vec{b}_1\| \leq \alpha \|\vec{b}_{d+1}^*\| \leq \alpha \beta \lambda_1(\lat)$
as needed.

Turning to Item~\ref{item:piling_HSVP}, we note that the HSVP conditions imply that $\|\vec{b}_1\|^d \leq \eta^{d(d-1)} \vol(\basis_{[1,d]})$ and $\|\vec{b}_{d+1}^*\|^{n-d} \leq \eta^{(n-d)(n-d-1)}\vol(\basis_{[d+1,n]})$. Using the bound on $\|\vec{b}_1\|$ relative to $\|\vec{b}_{d+1}^*\|$, we have
\[
\|\vec{b}_1\|^n \leq \eta^{2d(n-d)}\|\vec{b}_1 \|^d \cdot \|\vec{b}_{d+1}^*\|^{n-d} \leq \eta^{2(n-d)d + d(d-1) + (n-d)(n-d-1)} \vol(\basis) = \eta^{n(n-1)} \vol(\basis)
\; ,
\]
as needed.
\end{proof}

\subsection{The Micciancio-Walter DBKZ algorithm}
\label{sec:MW}

We  recall Micciancio and Walter's elegant 
DBKZ algorithm \cite{MWPracticalPredictable16}, as we will need it later. Formally, we slightly generalize DBKZ by allowing for the use of a $\delta$-SVP-oracle. We provide only a high-level sketch of the proof of correctness, as the full proof is the same as the proof in~\cite{MWPracticalPredictable16}, with Hermite's constant $\gamma_k$ replaced by $\delta^2 \gamma_k$.

\begin{algorithm}
	\small
	\caption{
		The Micciancio-Walter  DBKZ 
		algorithm \cite[Algorithm 1]{MWPracticalPredictable16}\label{alg:SDBKZ}}
	\begin{algorithmic}[1]
		\REQUIRE A block size $k \ge 2$, number of tours $N$,  
		a basis $\basis = (\mathbf{b}_{1}, \cdots,
		\mathbf{b}_{n}) \in \Z^{m \times n}$, and access to a $\delta$-SVP oracle for lattices with rank $k$.
		
		\ENSURE  A new basis of $\lat(\basis)$.
		
		\FOR{$\ell = 1$ \TO $N$}  
		\label{step of SDBKZ alg:repeat}
		
		\FOR{$i = 1$ \TO $n-k$}\label{step of SDBKZ alg:for}
		
		\STATE 
		$\delta$-SVP-reduce $\basis_{[i,i+k-1]}$.
		
		\label{step of LSDBKZ alg:size-reduction}
		\ENDFOR 
		\label{step of SDBKZ alg:firt end for}
		\FOR{$j = n-k+1$ \TO $1$} \label{step of SDBKZ alg:for2}
		
		\STATE 
		$\delta$-DSVP-reduce $\basis_{[j,j+k-1]}$
		
		\ENDFOR 
		\label{step of SDBKZ alg:end for}
		
		\ENDFOR
		\label{step of SDBKZ alg:until}
		\STATE  $\delta$-SVP-reduce $\basis_{[1,k]}$.
		\RETURN $\basis$. \label{step of SDBKZ alg:return}

	\end{algorithmic}
\end{algorithm}

\begin{theorem}\label{th:SDBKZ}
	For integers $n > k \geq 2$, an approximation factor $1 \leq \delta \leq 2^k$,
	an input basis $\basis_{0} \in \Z^{m \times n}$ for a lattice $\lat  \subseteq \Z^m$, and 
	$
	N := \lceil (2n^2/(k-1)^2) \cdot \log (n\log(5\|\basis_{0}\|)/\eps) \rceil$
	for some $\eps \in [2^{-\poly(n)},1]$, 
	Algorithm \ref{alg:SDBKZ} outputs a
	basis $\basis$ of $\lat$ in polynomial time (excluding oracle queries) such that
	\[
	\|\vec{b}_1\| \leq (1+\eps)\cdot (\delta^2 \gamma_{k})^{\frac{n-1}{2(k-1)}}\vol(\lat)^{1/n}
	\; 
	\]
	by making $N
	\cdot (2n-2k+1)+1$ calls to the $\delta$-SVP oracle for lattices with rank $k$.
\end{theorem}
\begin{proof}[Proof sketch]
We briefly sketch a proof of the theorem, but we outsource the most technical step to a claim from~\cite{MWPracticalPredictable16}, which was originally proven in~\cite{NeuBoundingBasis17}. Let $\basis^{(\ell)}$ be the basis immediately after the $\ell$th tour, and let 
$x_i^{(\ell)} := \log \vol(\basis_{[1,k+i-1]}^{(\ell)}) - \frac{k+i-1}{n} \log \vol(\lat)$ for $i=1,\ldots,n-k$. Let
	\[
	y_i := \frac{(n-k-i+1)(k+i-1)}{k-1} \cdot \log(\delta \sqrt{\gamma_k})\ \ \textrm{for}\ \ i=1,\ldots,n-k
	\; .
	\]
	By~\cite[Claim 3]{MWPracticalPredictable16} (originally proven in~\cite{NeuBoundingBasis17}), we have
	\[
	\max_{1\le i\le n-k} \big|x_i^{(\ell)}/y_i - 1\big| \leq (1-\xi) \max_{1\le i\le n-k} \big|x_i^{(\ell-1)}/y_i - 1 \big|
	\; ,
	\]
	where $\xi := 1/(1+n^2/(4k(k-1))) \geq 4(k-1)^2/(5n^2)$. Furthermore, notice that
	\[
	\max_{1\le i\le n-k} \big|x_i^{(0)}/y_i - 1 \big| \leq \frac{k(n-k)\log (5\|\basis^{(0)}\|)}{y_1}
	\;.
	\]
	It follows that 
	\begin{align*}
	\frac{x_1^{(N)}-y_1}{y_1}
	    &\leq (1-\xi)^N \max_{1 \le i\le n-k} \big| x_i^{(0)}/y_i - 1 \big| \\
	    &\leq  e^{-4(k-1)^2 N/(5n^2)}\cdot \frac{k(n-k)\log (5\|\basis^{(0)}\|)}{y_1}  \\
	    &\leq \frac{k\log(1+\eps)}{y_1}
	\;.
	\end{align*}
	
	In other words,
	\[
	\vol\big(\basis_{[1,k]}^{(N)}\big) 
	\leq (1+\eps)^{k} \cdot  (\delta^2 \gamma_{k})^{\frac{(n-k)k}{2(k-1)}}\vol(\lat)^{k/n}
	\; .
	\]
	Notice that the first vector $\vec{b}_1$ of the output basis is a $\delta$-approximate shortest vector in $\lat\big(\basis_{[1,k]}^{(N)}\big)$. Therefore,
	\[
	\|\vec{b}_1\| \leq \delta \sqrt{\gamma_k} \cdot \vol\big(\basis_{[1,k]}^{(N)}\big)^{1/k} \leq (1+\eps) (\delta^2 \gamma_k)^{\frac{n-1}{2(k-1)}} \vol(\lat)^{1/n}
	\; ,
	\]
	as needed.
	\end{proof}

\section{Slide reduction for \texorpdfstring{$n\le 2k$}{n at most 2k}}
\label{sec:bislide}

 In this section, we consider a generalization of Gama and Nguyen's slide reduction that applies to the case when $k < n \leq 2k$~\cite{GNFindingShort08}. Our definition in this case is not particularly novel or surprising, as it is essentially identical to Gama and Nguyen's except that our blocks are not the same size.\footnote{The only difference, apart from the approximation factor $\delta$, is that we use SVP reduction instead of HKZ reduction for the primal. It is clear from the proof in~\cite{GNFindingShort08} that only SVP reduction is required, as was observed in~\cite{MWPracticalPredictable16}. We \emph{do} require that additional blocks $\basis_{[i,n]}$ for $q+1\le i \leq k$ are SVP-reduced, which is quite similar to simply HKZ-reducing $\basis_{[q+1,n]}$, but this requirement plays a distinct role in our analysis, as we discuss below.}
 
 What \emph{is} surprising about this definition is that it allows us to achieve sublinear approximation factors for SVP when the rank is $n = k+q$ for $q = \Theta(k)$. Before this work, it seemed that approximation factors less than roughly $\gamma_q \approx n$ could not be achieved using the techniques of slide reduction (or, for that matter, any other known techniques with formal proofs). Indeed, our slide-reduced basis only achieves $\|\vec{b}_1\| \lesssim \gamma_q\lambda_1(\lat)$, which is the approximation factor resulting from the gluing lemma, Lemma~\ref{lem:pilingup}. (This inequality is tight.) We overcome this barrier by using our additional constraints on the primal together with some additional properties of slide-reduced bases (namely, Eq.~\eqref{eq:twin_volume}) to bound $\lambda_1(\lat(\basis_{[1,k]}))$. Perhaps surprisingly, the resulting bound is much better than the bound on $\|\vec{b}_1\|$, which allows us to find a much shorter vector with an additional oracle call.

\begin{definition}[Slide reduction]\label{def:slide-reduction}
Let $n =k+q$ where $1 \le q \le k$ 
are integers.
A basis $\basis$ of a lattice with rank $n$ is {\em $(\delta,k)$-slide-reduced} (with block size $k \geq 2$ and approximation
factor $\delta \ge 1$) if it is size-reduced and
satisfies the following set of conditions.
\begin{enumerate}
\item Primal conditions: The blocks  $\basis_{[1,q]}$ and $\basis_{[i,n]}$ for 
$i\in [q+1,\max\{k,q+1\}]$
 are $\delta$-SVP-reduced. 
\item Dual condition: the block $\basis_{[2,q+1]}$ is $\delta$-DSVP-reduced.
\end{enumerate}
\end{definition}

A reader familiar with the slide reduction algorithm from~\cite{GNFindingShort08} will not be surprised to learn that such a basis can be found (up to some small slack) using polynomially many calls to a $\delta$-SVP oracle on lattices with rank at most 
$k$. Before presenting and analyzing the algorithm, we show that such a slide-reduced basis is in fact useful for approximating SVP with sub-linear factors. (We note in passing that a slight modification of the proof of Theorem~\ref{th:property of bi-slide-reduction} yields a better result when $q = o(k)$. This does not seem very useful on its own, though, since when $q = o(k)$, the running times of our best SVP algorithms are essentially the same for rank $k$ and rank $k+q$.)

\begin{theorem}\label{th:property of bi-slide-reduction}
Let $\lat$ be a lattice with rank $n =k+q$ where $2 \le q \le k$ 
are integers. For any $\delta \ge 1$, 
if a basis $\vec{B}$ of 
$\lat$  
is $(\delta,k)$-slide-reduced, then,
\begin{equation*}\label{eq:bi-slide-reduction for local SVP}  
\lambda_{1}(\lat(\basis_{[1,k]})) 
\leq \delta\sqrt{\gamma_k}  (\delta^2 \gamma_{q})^{\frac{q+1}{q-1} \cdot \frac{n-k}{2k}} \lambda_{1}(\lat) 
  \;.
\end{equation*}
\end{theorem}
\begin{proof}
Let $\basis=(\mathbf{b}_{1}, \ldots, \mathbf{b}_{n})$. We distinguish two cases.

First, suppose that there exists an index $i \in [q+1,\max\{k,q+1\}]$ such that $\norm{\vec{b}_{i}^{\ast}} >\delta \lambda_{1}(\lat)$. Let $\vec{v}$ be a shortest non-zero vector of $\lat$. We claim that $\pi_i(\vec{v}) = 0$, i.e., that $\vec{v} \in \lat(\basis_{[1,i-1]})$. If this is not the case,  since $\basis_{[i,n]}$ is $\delta$-SVP-reduced, we have that
    \[
    \norm{\vec{b}_{i}^{\ast}}/\delta \le  \|\pi_i(\vec{v})\| \le  \|\vec{v}\| =\lambda_1(\lat),
\]    
which is a contradiction. Thus, we see that $\vec{v} \in \lat(\basis_{[1,i-1]}) \subseteq \lat(\basis_{[1,k]})$, and hence
 $\lambda_{1}(\lat(\vec{B}_{[1,k]}))=\lambda_{1}(\lat)$ (which is much stronger than what we need). 

Now, suppose that $\norm{\vec{b}_{i}^{\ast}} \le \delta \lambda_{1}(\lat)$ for all indices $i \in [q+1,\max\{k,q+1\}]$.
By definition, the primal and dual conditions imply that $\vec{B}_{[1,q+1]}$ is $\delta\sqrt{\gamma_q}$-twin-reduced.
Therefore, by Eq.~\eqref{eq:twin_volume} of Fact~\ref{fact:twinsies}, we have
 \begin{align*}
    \vol(\basis_{[1,k]}) &= \vol(\basis_{[1,q]}) \cdot \prod_{i=q+1}^{k}\norm{\vec{b}_{i}^{\ast}}\\
    &\le (\delta \sqrt{\gamma_{q}})^{q(q+1)/(q-1))} \|\vec{b}_{q}^{\ast}\|^{q}
    \cdot \prod_{i=q+1}^{k}\norm{\vec{b}_{i}^{\ast}}\\
    &\le (\delta^2 \gamma_{q})^{\frac{q+1}{q-1} \cdot \frac{n-k}{2}}
   (\delta\lambda_{1}(\lat))^{k}
   \; ,
  \end{align*}
where we have used 
the assumption that $\norm{\vec{b}_{i}^{\ast}} \le \delta \lambda_{1}(L)$ for all indices  $i \in [q+1,\max\{k,q+1\}]$ (and by convention we take the product to equal one in the special case when  $q = k$). 
By the definition of Hermite's constant, 
this implies that
\begin{equation*}
\lambda_{1}(\lat(\basis_{[1,k]})) \le \sqrt{\gamma_{k}}\vol(\basis_{[1,k]})^{1/k}\le \delta\sqrt{\gamma_k}  (\delta^2 \gamma_{q})^{\frac{q+1}{q-1} \cdot \frac{n-k}{2k}}
  \lambda_{1}(\lat) \;,    
\end{equation*}
as needed. 
\end{proof}

\subsection{The slide reduction algorithm for \texorpdfstring{$n \le 2k$}{n at most 2k}}
\label{subsec:Algorithm for bi-slide-reduction}

We now present our slight generalization of Gama and Nguyen's slide reduction algorithm that works for all $k+2 \leq n \leq 2k$.

\begin{algorithm}[h]
\small
\caption{The slide reduction algorithm for $n \le 2k$ (adapted from \cite[Algorithm 1]{GNFindingShort08})
\label{alg:BSR}}
\begin{algorithmic}[1]
\REQUIRE Block size $k$, slack $\eps > 0$, approximation 
factor $\delta \geq 1$,
a basis $\basis = (\mathbf{b}_{1}, \ldots, \mathbf{b}_{n}) 
\in \mZ^{m \times n}$ of a lattice $\lat$ 
with rank $n =k+q$ where $2\le q \le k$, and access to a $\delta$-SVP oracle for lattices with rank at most $k$.

\ENSURE  A $((1+\eps)\delta,k)$-slide-reduced basis of $\lat$.

\WHILE{$\vol(\vec{B}_{[1,q]})^{2}$ 
is modified by the loop} \label{BSR:while}

\STATE $\delta$-SVP-reduce 
$\basis_{[1,q]}$.

\FOR{$i = q+1$ \TO $\max\{k,q+1\}$} 
\STATE $\delta$-SVP reduce $\basis_{[i,n]}$.
\ENDFOR

\STATE Find 
a new basis $\vec{C} := (\vec{b}_1,\vec{c}_2,\ldots, \vec{c}_{q+1},\vec{b}_{q+2},\ldots,\vec{b}_n)$ of $\lat$ by $\delta$-DSVP-reducing $\basis_{[2,q+1]}$. 

\IF{$(1+\eps)\|\vec{b}_{q+1}^*\| < \| \vec{c}_{q+1}^*\|$}
\STATE $\basis \leftarrow \vec{C}$. \label{BSR:DSVP}
\ENDIF

\ENDWHILE\label{BSR:endwhiel}

\RETURN $\vec{B}$.\label{BSR:return}
\end{algorithmic}
\end{algorithm}

Our proof that Algorithm~\ref{alg:BSR} runs in polynomial time (excluding oracle calls)
is essentially identical to the proof in~\cite{GNFindingShort08}.
\begin{theorem} \label{th:correctness of bi-slide-reduction}
For $\eps \geq 1/\poly(n)$, Algorithm~\ref{alg:BSR} runs in polynomial time (excluding oracle calls), makes polynomially many calls to its $\delta$-SVP oracle, and outputs a $((1+\eps)\delta, k)$-slide-reduced basis of the input lattice $\lat$. 
\end{theorem}
\begin{proof}
First, notice that if Algorithm~\ref{alg:BSR} terminates, then its output must be 
$((1+\eps)\delta, k)$-slide-reduced. So, we only need to argue that the algorithm runs in polynomial time (excluding oracle calls).

Let $\basis_0 \in \Z^{m \times n}$ be the input basis and let $\basis \in \Z^{m \times n}$
denote the current basis during the execution of the algorithm. 
As is common in the analysis of basis reduction algorithms~\cite{LLLFactoringPolynomials82,GNFindingShort08,LNApproximatingDensest14}, 
we consider an integral potential of the form
\begin{equation*}
P(\basis) := \vol(\basis_{[1,q]})^{2}  \in \mZ^{+}
\; .\label{eq:potential for BSR}
\end{equation*}
The initial potential satisfies $\log P(\basis_{0}) \leq
2q \cdot \log \|\basis_{0}\|$, and every operation in Algorithm~\ref{alg:BSR} either
preserves or significantly decreases $P(\basis)$. More precisely, if the $\delta$-DSVP-reduction step (i.e., Step \ref{BSR:DSVP})  
occurs, 
then the potential $P(\basis)$ 
decreases by a multiplicative factor of at least  $(1+\varepsilon)^{2}$. No other step changes $\lat(\basis_{[1,q]})$ or $P(\basis)$.

Therefore, Algorithm~\ref{alg:BSR} updates $\lat(\basis_{[1,q]})$ at most $\frac{\log P(\basis_{0})}{2\log
(1+\varepsilon)}$ times, and hence it makes at most $\frac{qk\log \|\basis_{0}\|}{\log (1+\varepsilon)}$ calls to the $\delta$-SVP-oracle. From the complexity statement in Section \ref{subsec:lattice reduction},  
it follows that Algorithm~\ref{alg:BSR} runs efficiently (excluding the running time of oracle calls).
\end{proof}

\begin{corollary}
\label{cor:main1}
For any constant $c \in (1/2, 1]$ and $\delta := \delta(n) \geq 1$, there is an efficient reduction from $O(\delta^{2c+1} n^c)$-SVP on lattices with rank $n$ to $\delta$-SVP on lattices with rank $k := \lceil n/(2c) \rceil$.
\end{corollary}
\begin{proof}
On input (a basis for) an integer lattice $\lat \subseteq \Z^m$ with rank $n$,  the reduction first calls Algorithm~\ref{alg:BSR} to compute a $((1+\eps)\delta, k)$-slide-reduced basis $\basis$ of $\lat$  with, say, $\eps = 1/n$. The reduction then uses its $\delta$-SVP oracle once more on $\basis_{[1,k]}$ and returns the resulting nonzero short lattice vector. 

It is immediate from Theorem~\ref{th:correctness of bi-slide-reduction} that this reduction is efficient, and by Theorem~\ref{th:property of bi-slide-reduction}, the output vector is a $\delta'$-approximate shortest vector, where
\[
\delta' = \delta^2 \sqrt{\gamma_k}  ((1+\eps)^2\delta^2 \gamma_{q})^{\frac{q+1}{q-1} \cdot \frac{n-k}{2k}} \leq O(\delta^{2c+1} n^c)
\; ,
\]
as needed.
\end{proof}

\section{Slide reduction  for \texorpdfstring{$n\ge 2k$}{n at least 2k}}
\label{sec:Generalized-slide-reduction}

We now introduce a generalized version of slide reduction for lattices with any rank $n \geq 2k$. As we explained in Section~\ref{sec:techniques}, at a high level, our generalization of the definition from~\cite{GNFindingShort08} is the same as the original, except that (1) our first block $\basis_{[1,k+q]}$ is bigger than the others (out of necessity, since we can no longer divide our basis evenly into disjoint blocks of size $k$); and (2) we only $\eta$-HSVP reduce the first block (since we cannot afford to $\delta$-SVP reduce a block with size larger than $k$). Thus, our notion of slide reduction can be restated as ``the first block and the first dual block are $\eta$-(D)HSVP reduced and the rest of the basis $\basis_{[k+q+1,n]}$ is slide-reduced in the sense of~\cite{GNFindingShort08}.''\footnote{Apart from the approximation factor $\delta$, there is one minor difference between our primal 
conditions and those of~\cite{GNFindingShort08}. We only require the primal blocks to be SVP-reduced, while~\cite{GNFindingShort08} required them to be HKZ-reduced, which is a stronger condition. It is clear from the proof in~\cite{GNFindingShort08} that only SVP reduction is required, as was observed in~\cite{MWPracticalPredictable16}.} 

However, the specific value of $\eta$ that we choose in our definition below might look unnatural at first. We first present the definition and then explain where $\eta$ comes from.

\begin{definition}[Slide reduction]\label{def:GSR}
Let $n, k, p, q$ be integers such that $n = pk + q$ with $p,k \geq 2$ and $0 \le q \le k-1$, and let $\delta \geq 1$.
A basis $\basis \in \R^{m \times n}$ is {\em $(\delta,k)$-slide-reduced} if it is size-reduced and
satisfies the following three sets of conditions.
\begin{enumerate}
\item Mordell conditions: The block $\basis_{[1,k+q]}$ is $\eta$-HSVP-reduced and the block $\basis_{[2,k+q+1]}$ is $\eta$-DHSVP-reduced for $\eta := (\delta^2 \gamma_{k})^{\frac{k+q-1}{2(k-1)}}$.
\item Primal conditions: for all $i \in [1, p-1]$, the block 
$\basis_{[i k+q+1,(i+1)k+q]}$ is $\delta$-SVP-reduced.
\item Dual conditions: for all $i \in [1, p-2]$,  the block 
$\basis_{[ik+q+2,(i+1)k+q+1]}$ is $\delta$-DSVP-reduced.\footnote{When $p = 2$, there are simply no dual conditions.}
\end{enumerate}
\end{definition}

There are two ways to explain our specific choice of $\eta$. Most simply, notice that the output of the DBKZ algorithm---due to~\cite{MWPracticalPredictable16} and presented in Section~\ref{sec:MW}---is $\eta$-HSVP reduced when the input basis has rank $k+q$ (up to some small slack $\eps$). In other words, one reason that we choose this value of $\eta$ is because we actually can $\eta$-HSVP reduce a block of size $k+q$ efficiently with access to a $\delta$-SVP oracle for lattices with rank $k$. If we could do better, then we would in fact obtain a better algorithm, but we do not know how. Second, this value of $\eta$ is natural in this context because it is the choice that ``makes the final approximation factor for HSVP match the approximation factor for the first block.'' I.e., the theorem below shows that when we plug in this value of $\eta$, a slide-reduced basis of rank $n$ is $(\delta^2 \gamma_{k})^{\frac{n-1}{2(k-1)}}$-HSVP, which nicely matches the approximation factor of $\eta = (\delta^2 \gamma_{k})^{\frac{k+q-1}{2(k-1)}}$-HSVP that we need for the first block (whose rank is $k+q$). At a technical level, this is captured by Fact~\ref{fact:twinsies} and Lemma~\ref{lem:pilingup}.

Of course, the fact that these two arguments suggest the same value of $\eta$ is not a coincidence. Both arguments are essentially disguised proofs of Mordell's inequality, which says that $\gamma_n \leq \gamma_k^{(n-1)/(k-1)}$ for $2 \le k \le n$. E.g., with $\delta = 1$ the primal Mordell condition says that $\vec{b}_1$ yields a witness to Mordell's inequality for $\basis_{[1,k+q]}$.

\begin{theorem}\label{th:SR2}
For any $\delta \ge 1$, $k \geq 2$, and $n \geq 2k$,
if $\basis=(\vec{b}_{1}, \ldots, \vec{b}_{n}) \in \R^{m \times n}$ is a $(\delta,k)$-slide-reduced basis of a lattice $\lat$, then
\begin{equation}
\|\mathbf{b}_{1}\| \leq
(\delta^{2} \gamma_{k})^{\frac{n-1}{2(k-1)}}\vol(\lat)^{1/n} \;.
\label{eq:Gslide-reduction for HSVP}
\end{equation}
Furthermore, if $\lambda_{1}(\lat(\basis_{[1,k+q]}))> \lambda_{1}(\lat)$ , then
		\begin{equation}
		\|\mathbf{b}_{1}\| \leq
		\delta(\delta^2 \gamma_{k})^{\frac{n-k}{k-1}} \lambda_{1}(\lat) \; ,
		\label{eq:Gslide-reduction for ASVP}
		\end{equation}
		where $0 \leq q \leq k-1$ is such that $n=pk + q$.
\end{theorem}
\begin{proof}
Let $d := k + q$. Theorem~\ref{th:SR1} of Appendix~\ref{app:old_slide} shows that $\basis_{[d+1,n]}$ is both
$(\delta^{2} \gamma_{k})^{\frac{n-d-1}{2(k-1)}}$-HSVP-reduced
and $(\delta^{2} \gamma_{k})^{\frac{n-d-k}{(k-1)}}$-SVP-reduced. (We relegate this theorem and its proof to the appendix because it is essentially just a restatement of~\cite[Theorem 1]{GNFindingShort08}, since $\basis_{[d+1,n]}$ is effectively just a slide-reduced basis in the original sense of~\cite{GNFindingShort08}.)
Furthermore, $\basis_{[1,d+1]}$ is $(\delta^{2} \gamma_{k})^{\frac{d-1}{2(k-1)}}$-twin-reduced,
so that $\|\vec{b}_1\| \le (\delta^{2} \gamma_{k})^{\frac{d}{k-1}} \|\vec{b}_{d+1}^*\|$.
Applying Lemma~\ref{lem:pilingup} then yields both Eq.~\eqref{eq:Gslide-reduction for HSVP} and Eq.~\eqref{eq:Gslide-reduction for ASVP}.
\end{proof}

\subsection{The slide reduction algorithm for \texorpdfstring{$n \ge 2k$}{n at least 2k}}
\label{subsec:Algorithm for GSR}

We now present our slight generalization of Gama and Nguyen's slide reduction algorithm that works for all $ n \ge 2k$. Our proof that the algorithm  runs in polynomial time (excluding oracle calls) is essentially identical to the proof in~\cite{GNFindingShort08}.

\begin{algorithm}
\small
\caption{The slide-reduction algorithm for $n \ge 2k$ 
\label{alg:GSR}}
\begin{algorithmic}[1]
\REQUIRE Block size $k\ge 2$, slack $\eps > 0$, approximation factor $\delta \geq 1$,
basis $\basis = (\mathbf{b}_{1}, \ldots, \mathbf{b}_{n}) \in
\Z^{m \times n}$ of a lattice $\lat$ of rank $n =pk+q\ge 2k$ for $0\le q \le k-1$, and access to a $\delta$-SVP oracle for lattices with rank  $k$.

\ENSURE  A $((1 + \varepsilon)\delta,k)$-slide-reduced basis of $\lat(\basis)$.

\WHILE{$\vol(\basis_{[1,ik+q]})^{2}$ is modified by the loop for some 
$i \in [1, p-1]$}\label{GSR:while}

\STATE $(1+\eps)\eta$-HSVP-reduce $\basis_{[1,k+q]}$ 
using Alg.~\ref{alg:SDBKZ} for $\eta := (\delta^2 \gamma_{k})^{\frac{k+q-1}{2(k-1)}}$.\label{GSR:Mordell}

\FOR{$i = 1$ \TO $p-1$} 
 \STATE $\delta$-SVP-reduce 
$\basis_{[ik+q+1,(i+1)k+q]}$. 
\label{GSR:SVP}

\ENDFOR

\IF{$\basis_{[2,k+q+1]}$ is not $(1+\eps)\eta$-DHSVP-reduced} 

\STATE $(1+\eps)^{1/2} \eta$-DHSVP-reduce $\basis_{[2,k+q+1]}$ using 
Alg.~\ref{alg:SDBKZ}. \label{step:dual_mordell}

\ENDIF

\FOR{$i=1$ \TO $p-2$}

\STATE Find a new basis $\vec{C} := (\vec{b}_1,\ldots, \vec{b}_{ik+q+1},\vec{c}_{ik+q+2}, \ldots, \vec{c}_{(i+1)k+q+1},\vec{b}_{ik+q+2},\ldots,\vec{b}_n)$ of $\lat$ by $\delta$-DSVP-reducing $\basis_{[ik+q+2,(i+1)k+q+1]}$.

\IF{$(1+\eps) \|\vec{b}_{(i+1)k+q+1}^*\|< \| \vec{c}_{(i+1)k+q+1}^*\|$}

\STATE $\basis \leftarrow \vec{C}$. \label{step:GSR_DSVP}

\ENDIF

\ENDFOR 

\ENDWHILE\label{GSR:endwhile}

\RETURN $\basis$.
\end{algorithmic}
\end{algorithm}

\begin{theorem} \label{th:correctness of G-slide-reduction}
For $\eps \in [1/\poly(n),1]$, Algorithm~\ref{alg:GSR} runs in polynomial time (excluding oracle calls), makes polynomially many calls to its $\delta$-SVP oracle, and outputs a $((1+\eps)\delta, k)$-slide-reduced basis of the input lattice $\lat$. 
\end{theorem}
\begin{proof}
First, notice that if Algorithm~\ref{alg:GSR} terminates, then its output is $((1+\eps)\delta, k)$-slide-reduced. So, we only need to argue that the algorithm runs in polynomial time (excluding oracle calls).

Let $\basis_{0} \in \Z^{m \times n}$ be the input basis and let $\basis \in \Z^{m \times n}$
denote the current basis during the execution of Algorithm~\ref{alg:GSR}. 
As is common in the analysis of basis reduction algorithms \cite{LLLFactoringPolynomials82,GNFindingShort08,LNApproximatingDensest14}, 
we consider an integral potential of the form
\begin{equation*}
P(\basis) := \prod_{i=1}^{p-1} \vol(\basis_{[1,ik+q]})^{2}  \in \mZ^{+}.\label{eq:potential for GSR}
\end{equation*}
The initial potential satisfies $\log P(\basis_{0}) \leq 2n^{2}
\cdot \log \|\basis_{0}\|$, and every operation in Algorithm~\ref{alg:GSR} either
preserves or significantly decreases $P(\basis)$. In particular, the potential is unaffected by the primal steps (i.e., Steps~\ref{GSR:Mordell} and~\ref{GSR:SVP}), which leave $\vol(\basis_{[1,i k+q]})$ unchanged for all $i$. The dual steps (i.e., Steps~\ref{step:dual_mordell} and~\ref{step:GSR_DSVP}) either leave $\vol(\basis_{[1,i k+q]})$ for all $i$ or decrease $P(\basis)$ by a multiplicative factor of at least $(1+\eps)$.

Therefore, Algorithm~\ref{alg:BSR} updates $\vol(\basis_{[1,i k + q]})$ for some $i$ at most $\log P(\basis_{0})/\log
(1+\varepsilon)$ times. Hence, it makes at most $4pn^2  \log \|\basis_0\|/\log(1+\eps)$ calls to the SVP oracle in the SVP and DSVP reduction steps (i.e., Steps~\ref{GSR:SVP} and~\ref{step:GSR_DSVP}), and similarly at most $4n^2 \log \|\basis_0\|/\log(1+\eps)$ calls to Algorithm~\ref{alg:SDBKZ}. From the complexity statement in Section \ref{subsec:lattice reduction},  
it follows that Algorithm~\ref{alg:BSR} runs efficiently (excluding the running time of oracle calls), as needed.

\end{proof}

\begin{corollary}
\label{cor:main2}
For any constant $c \ge 1$ and $\delta := \delta(n) \geq 1$, there is an efficient reduction from $O(\delta^{2c+1}  n^c)$-SVP on lattices with rank $n$ to $\delta$-SVP on lattices with rank $k := \lfloor n/(c+1)\rfloor$. 
\end{corollary}
\begin{proof}
On input (a basis for) an integer lattice $\lat \subseteq \Z^m$ with rank $n$, the reduction first calls Algorithm~\ref{alg:GSR} to compute a $((1+\eps)\delta, k)$-slide-reduced basis $\basis=(\mathbf{b}_1, \cdots, \mathbf{b}_{n})$ of $\lat$  with, say, $\eps = 1/n$. Then, the reduction uses the procedure from Corollary~\ref{cor:main1} on the lattice $\lat(\basis_{[1,2k]})$ with $c = 1$ (i.e., slide reduction on a lattice with rank $2k$), to find a vector $\vec{v} \in \lat(\basis_{[1,2k]})$ with $0 < \|\vec{v}\| \leq O(\delta^{3}n) \lambda_1(\lat(\basis_{[1,2k]}))$.
Finally, the reduction outputs the shorter of the two vectors $\mathbf{b}_1$ and $\mathbf{v}$.

It is immediate from Corollary~\ref{cor:main1} and Theorem~\ref{th:correctness of G-slide-reduction} that this reduction is efficient. To prove correctness, we consider two cases. 

First, suppose that $\lambda_1(\lat(\basis_{[1,k+q]})) = \lambda_1(\lat)$.  Then, 
\[
\|\vec{v}\| \leq O(\delta^3 n) \lambda_1(\lat(\basis_{[1,2k]})) \leq O(\delta^{2c+1}  n^c) \lambda_1(\lat)
\; ,
\]
so that the algorithm will output a $O(\delta^{2c+1}  n^c)$-approximate shortest vector.

On the other hand, if $\lambda_1(\lat(\basis_{[1,k+q]})) > \lambda_1(\lat)$, then by Theorem~\ref{th:SR2}, we have
\[
\|\vec{b}_1\| \leq (1+\eps) \delta((1+\eps)^2\delta^2 \gamma_{k})^{\frac{n-k}{k-1}} \lambda_{1}(\lat) \leq O(\delta^{2c+1} n^c)
\; ,
\]
so that the algorithm also outputs a $O(\delta^{2c+1}  n^c)$-approximate shortest vector in this case.
\end{proof}

\bibliographystyle{alpha}

\appendix
\section{Properties of Gama and Nguyen's slide reduction}
\label{app:old_slide}

In the theorem below, $\basis_{[d+1,n]}$ is essentially just a slide-reduced basis in the sense of~\cite{GNFindingShort08}. So, the following is more-or-less just a restatement of~\cite[Theorem 1]{GNFindingShort08}.

\begin{theorem}\label{th:SR1}
Let $\basis = (\vec{b}_1,\ldots, \vec{b}_n) \in \R^{m \times n}$ with $n = pk + d$ for some $p \geq 1$ and $d \geq k$ be $(\delta,k)$-slide reduced in the sense of Definition~\ref{def:GSR}. Then,
\begin{align}
\|\mathbf{b}_{d+1}^*\| & \leq
(\delta^{2} \gamma_{k})^{ik/(k-1)} \|\mathbf{b}^*_{ik+d+1}\| \,\,\,\text{for}\,\,\, 0 \le i \le p-1 \;,
\label{eq:gap-SR1} \\
\|\mathbf{b}_{d+1}^*\| & \leq
(\delta^{2} \gamma_{k})^{\frac{n-d-1}{2(k-1)}}\vol(\basis_{[d+1,n]})^{1/(n-d)} \;, \text{ and}
\label{eq:HSVP-SR1} \\
\|\mathbf{b}_{d+1}\| & \leq
\delta(\delta^2 \gamma_{k})^{\frac{n-d-k}{k-1}} \lambda_{1}(\lat(\basis_{[d+1,n]})) \; .
\label{eq:ASVP-SR1}
    \end{align} 
\end{theorem}
\begin{proof}
By definition,
for each $i \in [0, p-2]$,
the block $\basis_{[ik+d+1,(i+1)k+d+1]}$ is $\delta \sqrt{\gamma_k}$-twin reduced.
By Eq.~\eqref{eq:twin_decay} of Fact~\ref{fact:twinsies}, we see that
\[
 \|\vec{b}_{(i-1)k+d+1}\| \le (\delta \sqrt{\gamma_k})^{2k/(k-1)} \|\vec{b}^*_{ik+d+1}\|
 \; ,
 \]
 which implies \eqref{eq:gap-SR1} by induction.
 
We prove \eqref{eq:HSVP-SR1} and  \eqref{eq:ASVP-SR1} by induction over $p$.
If $p=1$, then both inequalities hold as $\basis_{[d+1,n]}$ is $\delta$-SVP reduced by the definition of slide reduction.
Now, assume that Eqs.~\eqref{eq:HSVP-SR1} and~\eqref{eq:ASVP-SR1} hold for $p-1 \ge 1$.
Then $\basis$ satisfies the requirements of the theorem with $d' := d+k$ and $p' := p-1$. Therefore, by the induction hypothesis, we have
\begin{align*}
    \|\mathbf{b}_{d+k+1}^*\| & \leq
(\delta^{2} \gamma_{k})^{\frac{n-d-k-1}{2(k-1)}}\vol(\basis_{[d+k+1,n]})^{1/(n-d-k)} \;, \text { and}
 \\
\|\mathbf{b}_{d+k+1}\| & \leq
\delta(\delta^2 \gamma_{k})^{\frac{n-3k-\ell}{k-1}} \lambda_{1}(\lat(\basis_{[k+d+1,n]})) \; .
\end{align*}
Since $\basis_{[d+1,d+k]}$ is $\delta \sqrt{\gamma_k}$-HSVP reduced,
we may apply Lemma~\ref{lem:pilingup},
which proves \eqref{eq:HSVP-SR1} for $\basis_{[d+1,n]}$.

Furthermore, if $\lambda_1(\lat(\basis_{[d+1,n]})) < \lambda_1(\lat(\basis_{[d+1,d+k+1]}))$,
then $\basis_{[d+1,n]}$ is $\delta'$-SVP-reduced for
\[
\delta' = (\delta^{2} \gamma_{k})^{k/(k-1)} \cdot \delta(\delta^2 \gamma_{k})^{\frac{n-d-k}{k-1}} = \delta(\delta^2 \gamma_{k})^{\frac{n-d-k}{k-1}}
\; ,
\]
as needed.
If not, then $\lambda_1(\lat(\basis_{[d+1,n]})) = \lambda_1(\lat(\basis_{[d+1,d+k+1]}))$,
and $\|\vec{b}_1\| \le \delta \lambda_1(\lat(\basis_{[d+1,n]}))$
because $\basis_{[d+1,d+k+1]}$ is $\delta$-SVP reduced.
In all cases, we proved \eqref{eq:ASVP-SR1}.
\end{proof}

\end{document}